\documentclass{article}

\usepackage{graphicx}      
\usepackage{natbib}        

\usepackage{amsmath,amssymb,amsfonts,amsthm}
\usepackage{psfrag,graphicx,color}
\usepackage{enumerate}
\graphicspath{{./Figures/}}

\usepackage{dsfont}

\newtheorem{teo}{Theorem}
\newtheorem{theorem}[teo]{Theorem}
\newtheorem{lemma}[teo]{Lemma}
\newtheorem{prop}[teo]{Proposition}
\newtheorem{proposition}[teo]{Proposition}
\newtheorem{assumption}{Assumption}

    \newcommand{\setdef}[2]{\{#1 \, : \; #2\}}

\newcommand{\real}{\mathbb{R}}

\newcommand{\integernonnegative}{\mathbb{Z}_{\ge 0}}

\newcommand{\N}{\mathbb{N}}  

\newcommand{\G}{\mathcal{G}} 
\newcommand{\V}{\mathcal{V}} 
\newcommand{\E}{\mathcal{E}} 

\newcommand{\neigh}{ \mathcal{N}} 	
\newcommand{\card}[1]{|#1|}  	

\newcommand{\Exp}{\mathds{E}} 
\newcommand{\Prob}{\mathds{P}} 

\newcommand{\e}{\mathrm{e}} 

\newcommand{\1}{\mathbf{1}} 

\newcommand{\diag}{\operatorname{diag}} 

\title{Gossips and Prejudices: Ergodic Randomized Dynamics in Social Networks
\thanks{The authors are grateful to Profs. Giacomo Como, Fabio Fagnani, and Noah E. Friedkin for insightful conversations on the topics of this paper. }
}%

\author{Paolo Frasca\thanks{Department of Mathematical Sciences (DISMA), Politecnico di Torino, Torino, Italy. (e-mail: paolo.frasca@polito.it)} 
\and
Chiara Ravazzi\thanks{Department of Electronics and Telecommunications (DET), Politecnico di Torino, Italy. (e-mail: {chiara.ravazzi@polito.it})} 
\and
Roberto Tempo\thanks{CNR-IEIIT, Politecnico di Torino, Italy. (e-mail: roberto.tempo@polito.it)} 
\and
Hideaki Ishii\thanks{Department of Computational Intelligence and Systems Science, Tokyo Institute of Technology, Japan. (e-mail: {ishii@dis.titech.ac.jp})}
}%

\begin{document}
\maketitle

\begin{abstract}                
In this paper we study a novel model of opinion dynamics in social networks, which has two main features. First, agents asynchronously interact in pairs, and these pairs are chosen according to a random process. We refer to this communication model as ``gossiping''. Second, agents are not completely open-minded, but instead take into account their initial opinions, which may be thought of as their ``prejudices''. In the literature, such agents are often called ``stubborn".
We show that the opinions of the agents fail to converge, but persistently undergo ergodic oscillations, which asymptotically concentrate around a mean distribution of opinions. This mean value is exactly the limit of the synchronous dynamics of the expected opinions.
\end{abstract}

\section{Introduction}

The study of opinion dynamics has recently started to attract the attention of the control community. This interest is in large part motivated with the bulk of knowledge which has been developed about methods to approximate and stabilize consensus, synchronization, and other coherent states. 
However, in contrast with many engineering systems, social systems do not typically exhibit a consensus of opinions, but rather a persistence of disagreement, possibly with the formation of opinion parties.
It is then essential to understand which features of social systems prevent the formation of consensus. To the authors' understanding, scholars have focused on two key reasons: opinion-dependent limitations in the network connectivity and obstinacy of the agents.

The first line of research has seen a growth of models involving ``bounded confidence'' between the agents: if the opinions of two agents are too far apart, they do not influence each other. These models typically result into a clusterization of opinions: the agents split into non-communicating groups, and each group reaches an internal consensus. 
Influential models have been defined in~\cite{GD-DN-FA-GW:00,RH-UK:02}, and their understanding has been recently deepened by the control community, which has studied evolutions both in discrete time \cite{VDB-JMH-JNT:09,CC-FF-PT:12} and in continuous time~\cite{VDB-JMH-JNT:09a,FC-PF:11}, possibly including heterogenous agents~\cite{AM-FB:11f} and randomized updates~\cite{JZ-YH:12,JZ-GC-YH:12}.

Although interesting and motivated, these ``bounded confidence'' models do not seem to be sufficient to explain the persistence of disagreement in real societies, in spite of persistent contacts and interactions between agents. Instead, a persistent disagreement is more likely a consequence of the agents being unable, or unwilling, to change their opinions, no matter what the other agents' opinions are. This observation has been made {  by social scientists, as in the models introduced in~\cite{NEF-ECJ:99,NEF:11}, and more recently by physicists~\cite{MM-AP-SR:07}. Since this idea has spread to  applied mathematics and systems theory, several models have already been studied in detail,} using techniques from stochastic processes~\cite{GC-FF:10,DA-GC-FF-AO:11,EY-AD-AO-AS-AS:11} and from game theory~\cite{GCC-JSS:10,JG-RS:13}.

Following the latter line of research, in this paper we define a {\em{gossip dynamics}} such that at each time step a randomly chosen agent updates its opinion to a convex combination of its own opinion, the opinion of one of its neighbors, and its own initial opinion or ``prejudice''.
We show that, although the resulting dynamics persistently oscillates, its average is a stable opinion profile, which is not a consensus. { This means that the expected beliefs of an agent will not in general achieve, even asymptotically, an agreement  with the other agents in the society.
Furthermore, we show that the oscillations of opinions} are ergodic in nature, so that the averages along sample paths are equivalent to the ensemble averages. 

Our work has been deeply influenced from reading the papers~\cite{NEF-ECJ:99} and~\cite{DA-GC-FF-AO:11}, which also include agents with prejudices. Compared to the former paper, our contribution is a new model of communication between agents and, thus, of opinion evolution: a more precise discussion is given below in Section~\ref{sect:relation-gossip-friedkin}. Compared to the latter paper, which also proves an ergodic theorem, we allow the agents to have a continuum of degrees of obstinacy, rather than a dichotomy stubborn/non-stubborn. The qualitative picture, however, shows strong similarities. 

Finally, we point out that we have recently performed a similar analysis of ergodicity for a randomized algorithm, which solves the so-called localization problem for a network of sensors~\cite{CR-PF-HI-RT:13a}, \cite{CR-PF-RT-HI:13b}. We are confident that these techniques may foster the understanding of other randomized algorithms and dynamics~\cite{RT-GC-FD:12}, including for instance distributed PageRank computation~\cite{HI-RT:10}. 

\subsection*{Paper organization}
Sections~\ref{sect:friedkin-dynamics} and~\ref{sect:gossip-dynamics} are devoted to present the two models of opinion dynamics which we are interested in: the classical Friedkin and Johnsen's model and our new gossip algorithm, respectively. For both dynamics, we state a  convergence result. Section~\ref{sect:analysis} is then devoted to provide a proof of these statements. A few comments are given in the concluding section.

\subsection*{Notation}
Real and nonnegative integer numbers are denoted by $\real$ and $\integernonnegative$, respectively.
We use $\card{\mathcal{S}}$ to denote the cardinality of set $\mathcal{S}$, and $\|\cdot\|_2$ to denote the Euclidean norm. 
Provided $\G=(\V,\E)$ is a directed graph with node set $\V$ and edge set $\E$, we define for each node $i\in \V$ the set of neighbors $\neigh_i=\setdef{j\in \V}{(i,j)\in \E}$ and the degree $d_i=\card{\neigh_i}$. We assume that $(i,i)\in\mathcal{E}$ for all $i\in\V$, so that $d_i\ge1$ for every $i\in \V$. Such an edge is said to be a self-loop.
%
We refer the reader to~\cite{FB-JC-SM:09} for a broader introduction to graph theory and for related definitions.

\section{Social influence and prejudices}
\label{sect:friedkin-dynamics}
We consider two models of opinion dynamics: one is the well-known Friedkin and Johnsen's model, which we describe below, while the other is a related randomized model which we describe in the next section. 

We consider a set of agents $\V$, whose potential interactions are encoded by a directed graph $\G=(\V,\E)$, which we refer to as the {\em{social network}}. Each agent $i\in \V$ is endowed with a state $x_i(k)$, which evolves in discrete time, and represents its {\em{belief or opinion}}. We denote the vector of beliefs as $x(k)\in\real^\V$. An edge $(i,j)\in \E$ means that agent $j$ may directly influence the belief of agent $i$. To avoid trivialities, we assume that $\card{\V}>1$.

\subsection{Friedkin and Johnsen's model} 
Here we recall Friedkin and Johnsen's model and we give a convergence result based on the topology of the underlying social network.
%

Let $W\in \real^{\V\times\V}$ be a nonnegative matrix which defines the strength of the interactions ($W_{ij}=0$ if $(i,j)\not\in \E$) and $\Lambda$ be a diagonal matrix describing how sensitive each agent is to the opinions of the others, based on interpersonal influeneces. We assume that $W$ is row-stochastic, i.e., $W\1=\1$, where $\1$ denotes the vector of ones, and we set $\Lambda=I-\diag(W)$, where $\diag(W)$ collects the self-weights given by the agents.
The dynamics of opinions $x(k)$ proposed in \cite{NEF-ECJ:99} is 
\begin{equation}\label{eq:friedkin}x(k+1)= \Lambda W x(k) + (I-\Lambda) u,\end{equation}
with $x(0)=u$ and $u\in \real^\V$. The vector $u$, which corresponds to the individuals' preconceived opinions, also appears as an input at every time step. The presence of this input is the main feature of this model, and marks its difference with, for instance, the mentioned models which are based on bounded confidence. 
As a consequence of~\eqref{eq:friedkin}, the opinion profile at time $k\in \integernonnegative$ is equal to 
$$x(k)=\big ( (\Lambda W)^k +\sum_{s=0}^{k-1} (\Lambda W)^s (I-\Lambda) \big) u.$$ 
The limit behavior of the opinions is described in the following result.

\begin{proposition}[Convergence]\label{prop:convergence-friedkin}
Assume that from any node $\ell\in \V$ there exists a path from
$\ell$ to a node $m$ such that  $W_{mm}>0$.
Then, the opinions converge and
$$x':=\lim_{k\to+\infty} x(k)=(I-\Lambda W)^{-1}(I-\Lambda)u.$$
\end{proposition}

\begin{proof}
 Due to the assumption, $\Lambda$ is a substochastic matrix, that is, a matrix with positive entries which sum to less than or equal to one along each row.
Then, $\Lambda W$ is substochastic also, and 
Schur stable { by Lemma~\ref{lemma:substoch_stab} (proved in Section~\ref{sect:analysis})}. Thus, the dynamics in (1) with the constant input
term $(I-\Lambda)u$ is convergent to $x'$. 
\end{proof}
 
The assumption of the proposition implies that each agent is influenced by at least one stubborn agent.
As shown in the proof, this is sufficient to guarantee the stability of the opinion dynamics.
In practice, we expect that in a social network most agents will have some level of obstinacy, thus a positive $W_{ii}$.

Let $V:=(I-\Lambda W)^{-1}(I-\Lambda)$, which is referred to as
the total effects matrix in~\cite{NEF-ECJ:99}.
Since $W$ is stochastic,
we observe that under the assumption of
Proposition~\ref{prop:convergence-friedkin} also $V$ is  stochastic:
this means that the limit opinion of each agent is a convex combination
of the preconceived  opinions of the group. However, we note that 
$x'$ is not a ``consensus'',
but a more general opinion profile such that
$x_i'=\sum_j V_{ij}u$.
Note that instead a consensus is reached if $W$ has zero
diagonal ({\it i.e.}, $\Lambda=I$) and the graph {  is aperiodic and} has a globally
reachable node; see for instance~\cite{FB-JC-SM:09}.

\subsection{Example from~\cite{NEF-ECJ:99}}\label{sect:example}

Here, we briefly describe an example from~\cite{NEF-ECJ:99} to illustrate
how the opinion dynamics arise in the context of social networks.

We consider a group of $N$ agents and study how opinions are formed
through interactions. The model is an abstraction of experiments
conducted and reported in the reference. The general flow
of the experiments is as follows:
\begin{enumerate}
\item The agents are presented with an issue
(related to sports, surgery, school, etc.)
on which opinions can range, say, from 1 to 100.
\item Each agent forms an initial opinion on the issue.
\item The agents can communicate over phone with other
agents (predetermined by the experiment organizer)
individually to discuss the issue. 
\item After several rounds of discussion, they settle on final
opinions that may or may not be in agreement.
\item They are also asked to provide estimates of the relative
interpersonal influences of other group members on
their final opinions. 
\end{enumerate}

As a simple example, we describe a case with four agents.
Let the initial and final opinion vectors be
\begin{align*}
  x(0) &= [25~ 25~ 75~ 85]^{\top},\\
  x' &=  [60~ 60~ 75~ 75]^{\top}.
\end{align*}

The matrix $W$ which determines the influence network for this
group is given by\footnote{ Note that the matrix $W$ slightly differs from that on page six of~\cite{NEF-ECJ:99} because the rows of the latter do not sum exactly to one, due to rounding errors.}
\[
 W = \begin{bmatrix}
       .220 & .120 & .360 & .300\\
       .147 & .215 & .344 &.294\\
        0   & 0   & 1   & 0\\
     .090   & .178 & .446 &.286
    \end{bmatrix},
%
\]
where the entries represent the distribution of relative
interpersonal influences on the issue. {  Note that agent 3 in this example is ``totally stubborn'', meaning that it does not change its opinion at all during the evolution.} This matrix is obtained
from the experiment data $x(0)$, $x'$, and the estimate of
the relative interpersonal influences.

We take $\Lambda =I-\text{diag}(W)=\text{diag}(.780,.785,0,.714)$;
the entries represent the agents' susceptibilities to interpersonal
influence. The off-diagonal entries of $W$ are the weights
of the influence by the others.
For example, $W_{12}=.12$ shows that the direct relative
influence of agent~2 on agent~1 is .120. The matrix $V$ is
\[
 V = \begin{bmatrix}
       .280 & .045 & .551 & .124\\
       .047 & .278 & .549 &.126\\
        0   & 0   & 1   & 0\\
     .030   & .048 & .532 &.390
    \end{bmatrix}.
\]
This matrix indicates the influence of each agent on every other
agent in the final opinions through the flow of direct and indirect
interactions. For example, $V_{23}=.549$ shows that almost 55\%
of the final opinion of agent~2 is determined by agent~3.

The evolution of the opinions is illustrated by the simulations in Figure~\ref{fig:FJmodel}, which respectively plot the state $x(k)$ and the corresponding limit point $x'$ (marked by blue circles).

\begin{figure}[htb]
\begin{center}
\psfrag{x}{$x$}
\psfrag{time}{$k$}
\psfrag{averaged}{$\overline{x}$}
\includegraphics[width=.6\columnwidth]{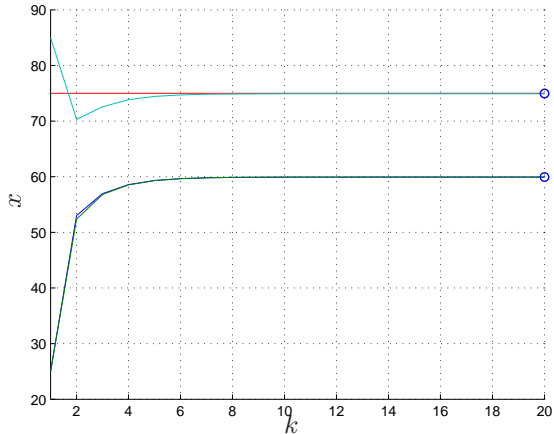}
\caption{Evolution of the opinion dynamics~\eqref{eq:friedkin} in the example of Section~\ref{sect:example}. The opinions $x$ converge to the limit $x'$ (marked by blue circles).}
\label{fig:FJmodel}
\end{center}
\end{figure}

\section{Gossip opinion dynamics}
\label{sect:gossip-dynamics}
We propose here a class of randomized opinion dynamics, which translates the idea of Friedkin and Johnsen's model into a ``gossip'' communication model.

The agents' beliefs evolve according to the following stochastic update process.
Each agent $i\in \V $ starts with an initial belief $x_i(0)=u_i\in \real$. At each time $k\in\integernonnegative$ a directed link is randomly sampled from a uniform distribution over $\E$. If the edge $(i, j)$ is selected at time $k$, agent $i$ meets agent $j$ and updates its belief to a convex combination of its previous belief, the belief of $j$, and its initial belief. Namely, 
\begin{align}\label{eq:gossip-friedkin}
\nonumber x_i(k+1)&=h_{i}\big((1-\gamma_{ij})x_i(k)+\gamma_{ij}x_j(k)\big)+(1-h_{i})u_i\\
x_\ell(k+1)&=x_\ell(k)\qquad \forall \ell\in \V\setminus\{i\},
\end{align}
where the weighting coefficients $h_i$ and $\gamma_{ij}$ satisfy the following assumption.
\begin{assumption}\label{assmp:coefficients}\rm
Let the diagonal matrix $H$ be defined by $H_{ii}=h_i$ and the matrix $\Gamma$ defined by $\Gamma_{ij}=\gamma_{ij}$. We assume that (i) $h_i\in [0,1]$ for all $i\in \V$; (ii) $\Gamma$ is row-stochastic, {\it i.e.}, for all $i$ and $j$ in $\V$ it holds $\gamma_{ij}\ge0$, $\sum_\ell \gamma_{i\ell}=1$; and (iii) $\gamma_{ij}=0$ if $j$ is not a neighbor of $i$.
\end{assumption}

As a consequence of this assumption, we observe that at all times the opinions of the agents are convex combinations of their initial prejudices.
Note that if an edge of the form $(i,i)$ is sampled at time $k$, then
$$ x_i(k+1)=h_{i} x_i(k)+(1-h_{i})u_i,
$$
that is, the opinion of agent $i$ moves back closer to its preconceived opinion.
{  Also note that if $h_i=0$, then agent $i$ is totally stubborn, whereas if $h_i=1$, then agent $i$ is totally open-minded: we may say that $1-h_i$ is proportional to the obstinacy of the agent.}

Our analysis -detailed in the next section- requires to study first the average dynamics of the gossip model.
\begin{lemma}[Expected dynamics]\label{lemma:expected-dynamics}
Under Assumption~\ref{assmp:coefficients}, the dynamics~\eqref{eq:gossip-friedkin} is such that
\begin{align}
\label{eq:mean-dynamics-gossip-opinions}
\Exp[x(k+1)]
=&\Big(I-\frac{1}{\card{\E}}\big(D(I-H) +H(I-\Gamma)\big)
\Big)\Exp[x(k)] 
+\frac{1}{|\E|}D(I-H)u,\end{align}
where $D$ is the degree matrix of $\G$, a diagonal matrix whose diagonal entry is equal to the degree $d_i=\card{\neigh_i}$.
\end{lemma}

Based on this preliminary result, we make the following statement about the convergence properties of the gossip opinion dynamics. 
In order to prove it, we also make an assumption involving the topology of the network as follows. This assumption corresponds to that in Proposition~\ref{prop:convergence-friedkin} on the diagonal entries of $W$.
\begin{assumption}\label{assmp:hm}\rm
From each node $\ell$ in $\V$, there exists
a path in $\G$ from $\ell$ to a node $m$ such that $h_m\neq 1$. 
\end{assumption}
{ 
Note that if $\G$ is strongly connected, then Assumption~\ref{assmp:hm} is satisfied.}

\begin{theorem}[Ergodicity and limit behavior]\label{thm:gossip-opinions} 
Under Assumptions~\ref{assmp:coefficients} and~\ref{assmp:hm}, it holds that
\begin{enumerate}
\item the expected dynamics~\eqref{eq:mean-dynamics-gossip-opinions} converges and
\begin{align*}
x^{\star}&:=\lim_{k\rightarrow\infty}\Exp[x(k)]\\
&=(D(I-H)+H(I-\Gamma))^{-1}D(I-H)u;\end{align*}
\item the dynamics~\eqref{eq:gossip-friedkin} is {\em mean-square-ergodic}, that is, 
\begin{align}\label{eq:gossip-is-ergodic}\lim_{k\to+\infty}\Exp[\|{\bar x}(k)-x^{\star} \|_2^2]=0,\end{align}
where $$\bar x(k)=\frac1{k+1}\sum_{\ell=0}^k x(\ell).$$
\end{enumerate}
\end{theorem}
\smallskip
 We note here that the assumption of   having $H$ different from the identity matrix is not restrictive. Indeed, if instead $H=I$ and $\G$ has a globally reachable node, then the dynamics reduces to a standard asymmetric gossip consensus algorithm, studied for instance in~\cite{FF-SZ:08b,FF-SZ:08a}.
We also notice that our results assume the edges to be chosen for the update according to a uniform distribution. This choice is made for simplicity, but our analysis can easily be extended to consider more general or different distributions.

In our result, we prove ergodicity in the sense that the time-averages (also known as Ces\`aro averages or Polyak averages in some contexts) converge to the limit of the expected dynamics in mean square sense. We could also have given the corresponding statement of almost-sure convergence: its proof would closely follow the arguments in~\cite{DA-GC-FF-AO:11}.
In this work, we prefer to focus on mean square convergence because we are able to give a detailed proof by a direct argument.
The ergodicity of the opinion dynamics is illustrated by the simulations in Figure~\ref{fig:demo}, which respectively plot the state $x(k)$ and the corresponding time-averages.

\begin{figure}[htb]
\begin{center}
\psfrag{x}{$x$}
\psfrag{time}{$k$}
\psfrag{averaged}{$\overline{x}$}
\includegraphics[width=0.49\columnwidth]{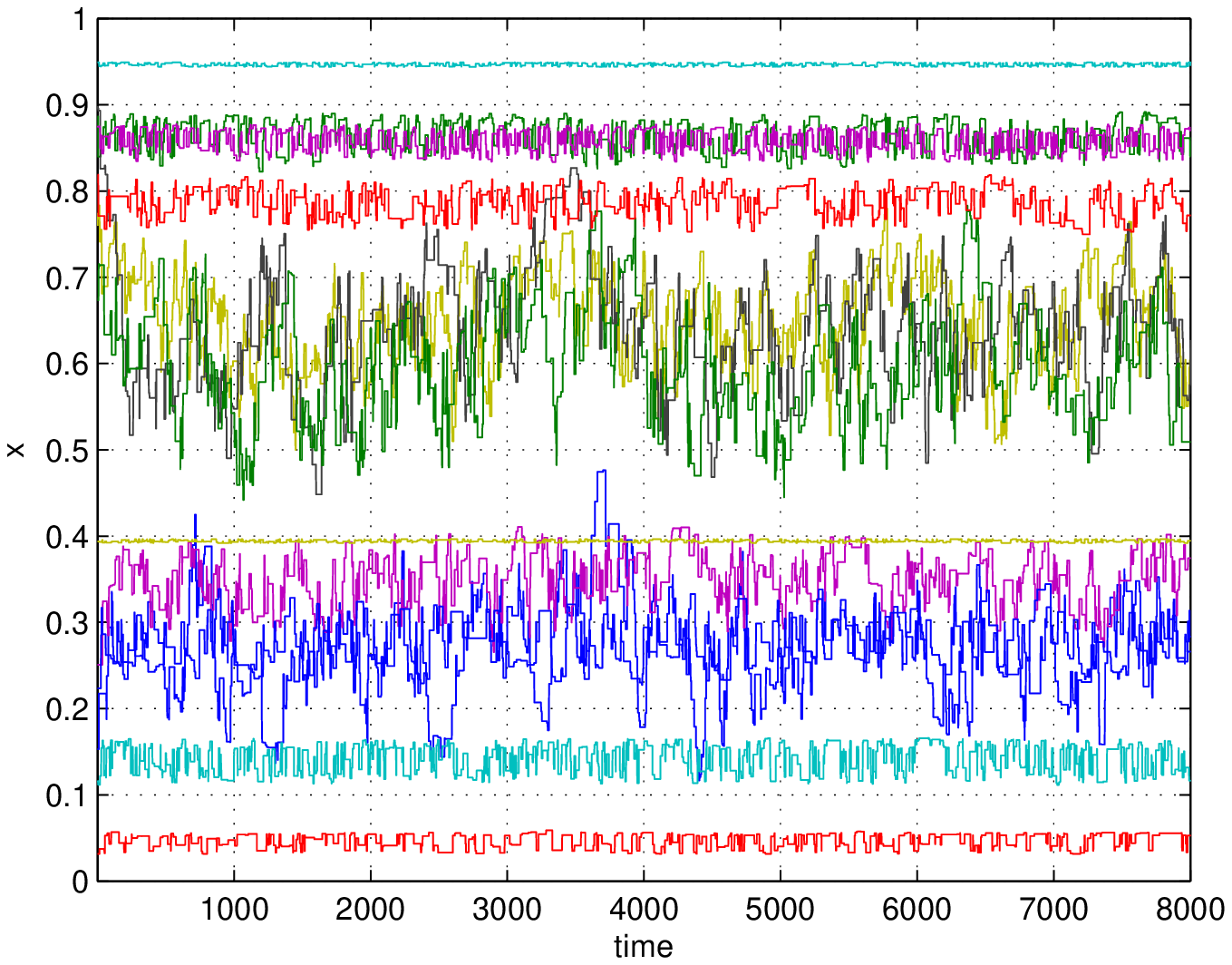}
\includegraphics[width=0.49\columnwidth]{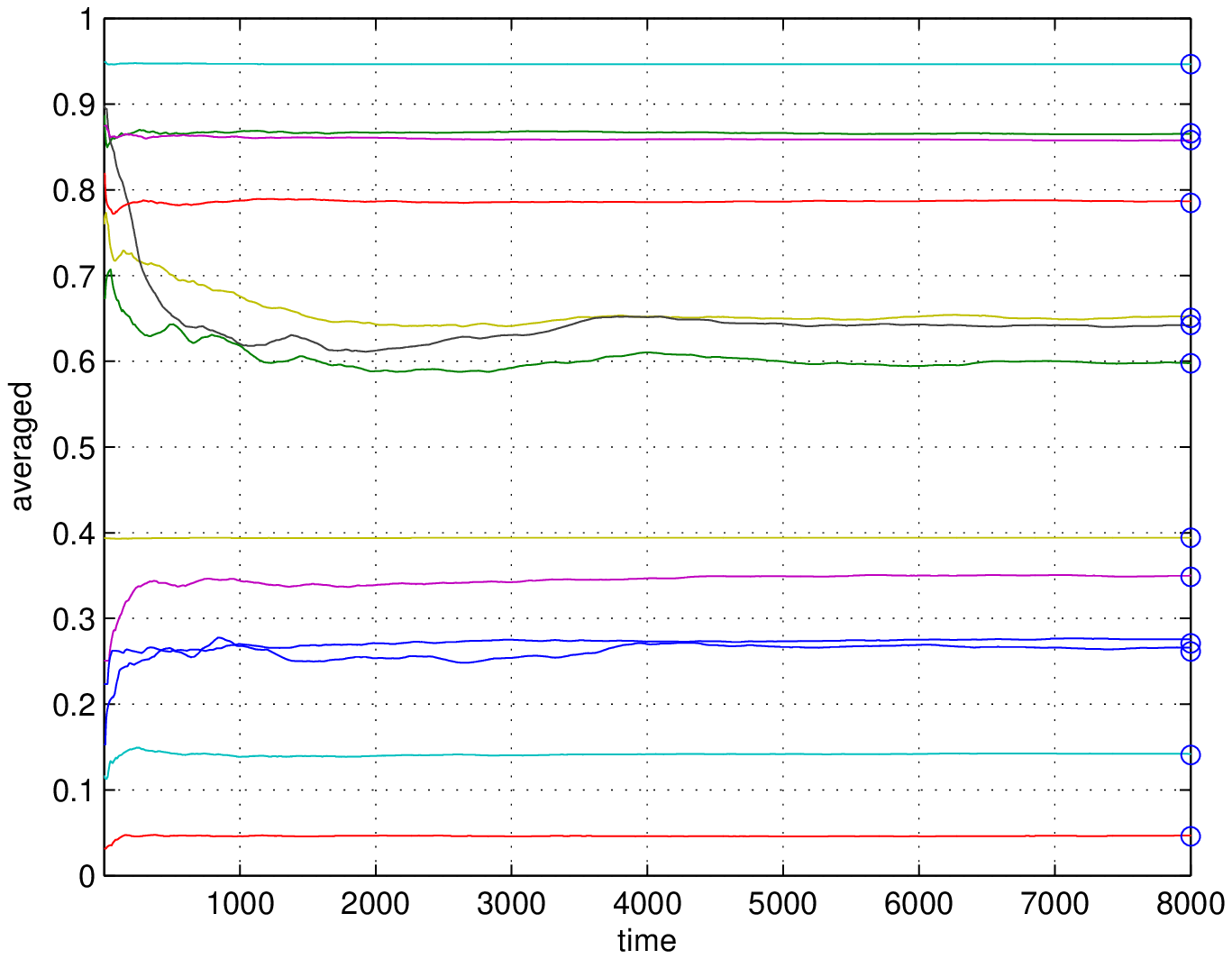}
\caption{Typical sample-path behavior of the beliefs in a social network topology with $n=13$. The belief process $x$ (top plot) oscillates persistently in a bounded interval. As the belief process is ergodic, the time averages (bottom plot) converge, when time goes to infinity, to $x^\star$ (marked by blue circles).}
\label{fig:demo}
\end{center}
\end{figure}

\subsection{Relating gossip and synchronous dynamics}
\label{sect:relation-gossip-friedkin}

We now discuss the interpretation of our convergence theorem in the
context of opinion dynamics. The original proposal by Friedkin and Johnsen abstracts from a precise analysis of the communication process among the agents, and postulates synchronous rounds of interaction. In fact, the lack of a more precise model for inter-agent interactions is acknowledged in~\cite{NEF-ECJ:99} by saying that ``it is obvious that interpersonal influences do not occur in the simultaneous way that is assumed''.

The proposed gossip dynamics tries to introduce a more reasonable model of  the communication process among the agents. Indeed, as mentioned in the example in the previous section, in the experiments conducted in  \cite{NEF-ECJ:99}, the agents were allowed to discuss pairwise. 

In what follows, we investigate the relationship between the two dynamics more carefully.
From a purely mathematical point of view, we observe
the following fact.
\begin{prop}\label{prop:relations-with-friedkin} If $H$ and $\Gamma$ are chosen as
\begin{equation} \label{eqn:H}
 h_i=\begin{cases}
({d_i-(1-\lambda_{ii})})/{d_i}&\text{if }d_i\neq1\\
{ 0}&\text{otherwise}
\end{cases}
\end{equation}
\begin{equation}\label{eqn:Gamma}
 \gamma_{ij}=\begin{cases}
 \frac{d_i(1-h_i)+h_i-(1-\lambda_{ii}w_{ii})}{h_i}&\text{if }i=j,\ d_i\neq1\\
 \frac{\lambda_{ii}w_{ij}}{h_i}&\text{if }i\neq j,\ d_i\neq 1\\
 1&\text{if }i=j,\ d_i=1\\
 0&\text{if }i\neq j,\ d_i=1
 \end{cases}
\end{equation}
then $\Gamma$ and $H$ satisfy Assumption~\ref{assmp:coefficients} and the expected dynamics \eqref{eq:mean-dynamics-gossip-opinions} can be written as
\begin{equation}\label{eq:mean-dynamics-gossip-opinions-b}\Exp[x(k+1)]
=\big(I-\frac{1}{\card{\E}}) \Exp[x(k)] + \frac{1}{\card{\E}} \big(\Lambda W \Exp[x(k)] +(I-\Lambda)u\big). \end{equation} 
Furthermore, $x^\star=x'$.
\end{prop}
\begin{proof}
 First, we notice trivially that $h_i\in[0,1]$, and since if $d_i=1$ then $\gamma_{ij}\in\{0,1\}$ with $\sum_{j}\gamma_{ij}=1$.
If $i\in\V$ is such that $d_i>1$, then $\gamma_{i,j}\geq0$ and
\begin{align*}
\gamma_{ii}\geq0 \iff d_i(1-h_i)+h_i-(1-\lambda_{ii}w_{ii})\geq0,
\end{align*}
from which we observe
\begin{align*}
d_i(1-h_i)+h_i-(1-\lambda_{ii}w_{ii})&\geq d_i(1-h_i)+h_i-1\\
&\geq\min\{d_i-1,0\}=0.
\end{align*}
We deduce that all the entries of $\Gamma$ are nonnegative.
Let us compute now
\begin{align*}
\sum_{j\in \V}\gamma_{ij}
&=\frac{1}{h_i}\big(\lambda_{ii}\sum_{j\neq i}{w_{ij}}+d_i(1-h_i)+h_i-1+\lambda_{ii}w_{ii}\big)\\
&=\frac{1}{h_i}\left[\lambda_{ii}({1-w_{ii}})+d_i(1-h_i)+h_i-1+\lambda_{ii}w_{ii}\right]\\
&=\frac{1}{h_i}\left[\lambda_{ii}+1-\lambda_{ii}+h_i-1\right]=1.
\end{align*}
We conclude that $\Gamma$ is row-stochastic and has all entries in the interval [0,1].
The thesis is then obtained by noticing that the expressions in \eqref{eqn:H} and \eqref{eqn:Gamma} imply
$$
D(I-H)=I-\Lambda
$$
and 
$$
D(I-H)+H(I-\Gamma)=I-\Lambda W.
$$
\end{proof}
 
In words, we may say that, under the assumption that $\Gamma$ and $H$ are chosen as in Proposition~\ref{prop:relations-with-friedkin}, then the expected
dynamics~\eqref{eq:mean-dynamics-gossip-opinions-b} is a ``lazy'' (slowed down) version of the Friedkin and Johnsen's dynamics associated to the matrix $W$. Hence, Theorem \ref{thm:gossip-opinions} shows that the average dynamics $\Exp[x(k)]$ converges to the { limit} opinions of the original model~\eqref{eq:friedkin}.
This relationship between the two dynamics provides an additional justification and a new perspective on the model originally proposed by Friedkin and Johnsen.
Furthermore, we observe that Proposition~\ref{prop:convergence-friedkin} can be immediately deduced as a corollary of Theorem~\ref{thm:gossip-opinions}. 

The form of~\eqref{eqn:H} and \eqref{eqn:Gamma} may seem complicated at first sight. However, this is not surprising if we think of other examples of randomized dynamics over networks. For instance, in problems of consensus~\cite{FF-SZ:08a}, localization~\cite{CR-PF-HI-RT:13a}, and PageRank computation~\cite{HI-RT:10}, the definition of the update matrices of the randomized dynamics is not trivial and must be done carefully in order to reconstruct, on average, the desired synchronous dynamics.

\subsection{Example (continued)}\label{sect:example-gossip}

In this subsection, we continue with the example of four agents.
Let the weight matrices $\Gamma$ and $H$ in the update equation (2)
be chosen according to \eqref{eqn:H} and \eqref{eqn:Gamma}. Then we have
\begin{align*}
 H &=\text{diag}( .945,0.946,.000,.928),\\
 \Gamma
  &= \begin{bmatrix}
 .356  &  .099  & .297  &  .248\\
    .122  &  .349  &  .285  &  .244\\
         0    &     0   & 1.0000     &    0\\
    .069   & .137 &  .343 &   .451
     \end{bmatrix},\\
 D &=\text{diag}(4,4,1,4).
\end{align*}
Suppose that at time $k$, as a result of gossiping, the edge between agents 1 and 2 is
chosen. In this case, the dynamics (2) can be written in the
matrix form as follows:
\begin{align*}
 x(k+1)
 &= \begin{bmatrix}
     h_1 (1-\gamma_{12}) & h_1\gamma_{12} & 0 & 0\\
     0 & 1 & 0 & 0\\
     0 & 0 & 1 & 0\\
     0 & 0 & 0 & 1
    \end{bmatrix} x(k)+ \begin{bmatrix}
       1-h_1 & 0 & 0 & 0\\
       0     & 0 & 0 & 0\\
       0 & 0 & 0 & 0\\
       0 & 0 & 0 & 0
     \end{bmatrix}u\\
 &= \begin{bmatrix}
       .851 &    .094 & 0 & 0\\
      0 &  1 & 0 & 0\\
     0 & 0 & 1 & 0\\
     0 & 0 & 0 & 1
    \end{bmatrix} x(k)
   + \begin{bmatrix}
         .055 & 0 & 0 & 0\\
       0     &  0 & 0 & 0\\
       0 & 0 & 0 & 0\\
       0 & 0 & 0 & 0
     \end{bmatrix}u.
\end{align*}
For other edge choices, similar expressions can be obtained.
We can see in simulations (Figure~\ref{Fig1andFig2}) that the states oscillate, but the time averages converge, as predicted by our results.
\begin{figure}[htb]
\begin{center}
\psfrag{x}{$x$}
\psfrag{time}{$k$}
\psfrag{averaged}{$\overline{x}$}
\includegraphics[width=0.49\columnwidth]{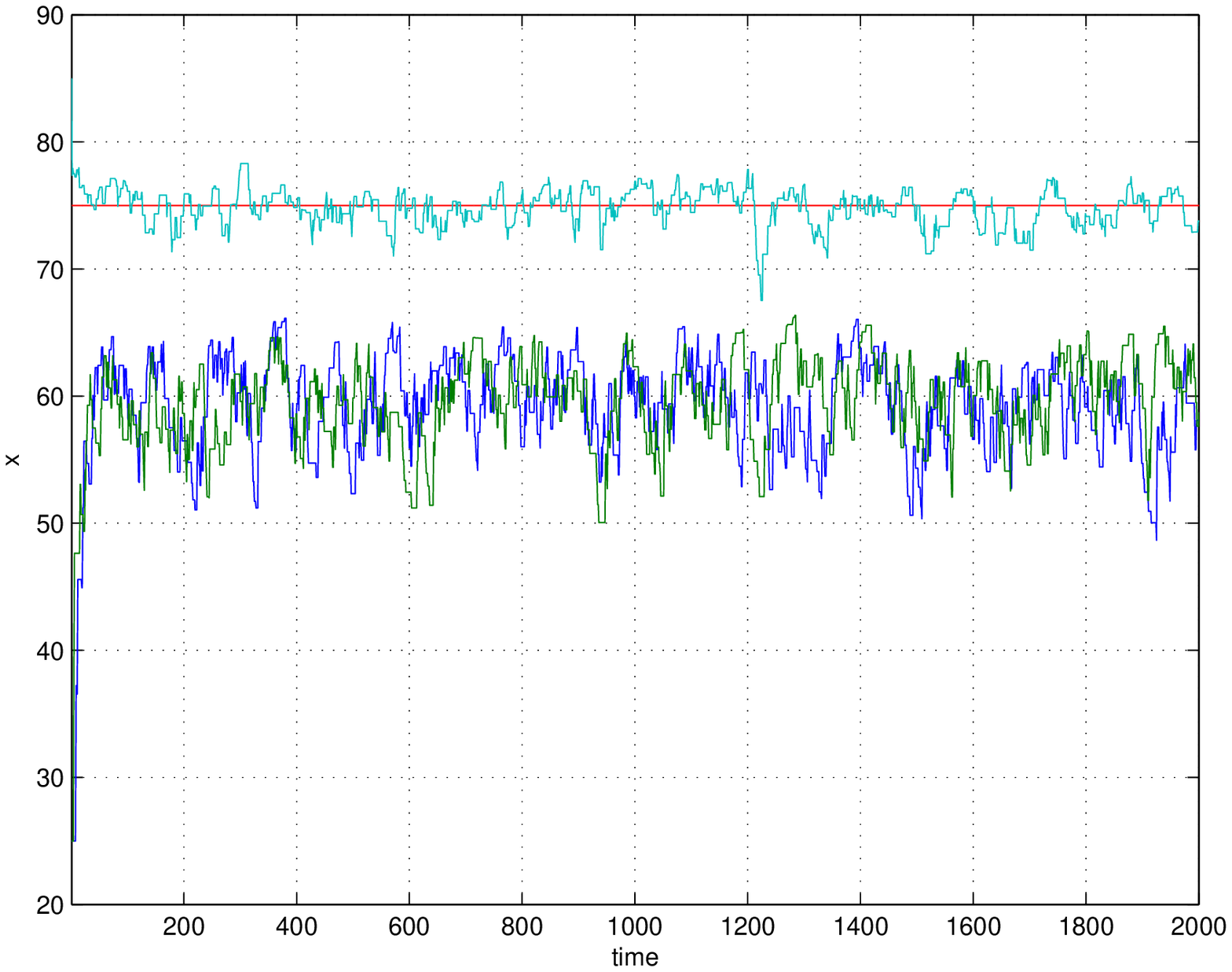}
\includegraphics[width=0.49\columnwidth]{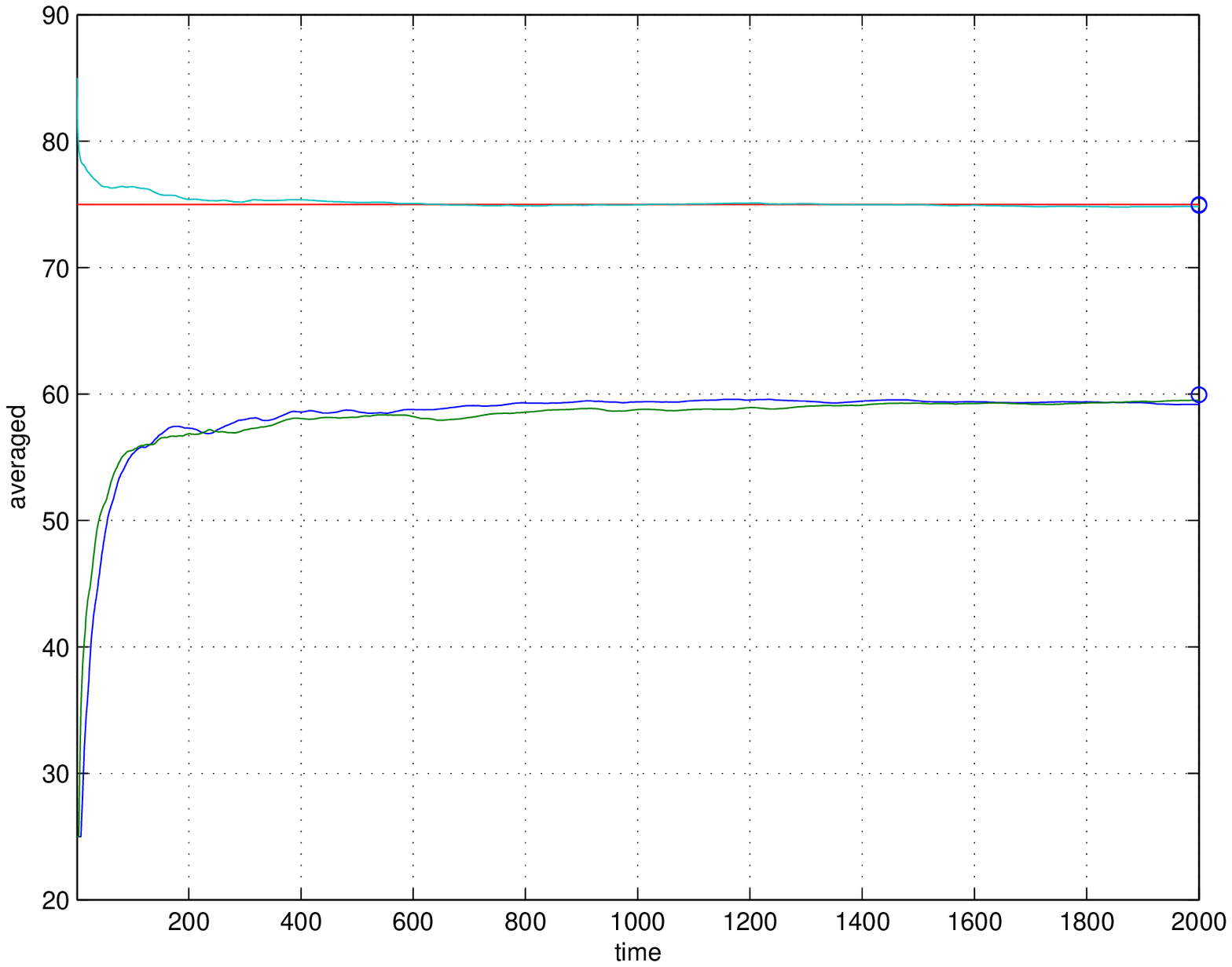}
\caption{Evolution of the opinions in the example of Section~\ref{sect:example-gossip}. The belief process $x$ (top plot) oscillates persistently in a bounded interval. As the belief process is ergodic, the time averages (bottom plot) converge, when time goes to infinity, to the limit of Friedkin and Johnsens' dynamics (marked by blue circles).}
\label{Fig1andFig2}
\end{center}
\end{figure}

\section{Analysis of the gossip model}
\label{sect:analysis}
This section is devoted to prove Lemma~\ref{lemma:expected-dynamics} and Theorem~\ref{thm:gossip-opinions}.
Their proof requires a few steps: (i) we rewrite~\eqref{eq:gossip-friedkin} as a random affine iterate; (ii) we compute relevant moments of the matrices involved in the affine iterate; (iii) we verify that the mean dynamics~\eqref{eq:mean-dynamics-gossip-opinions} is stable and we compute its limit; (iv) we prove ergodicity by verifying~\eqref{eq:gossip-is-ergodic} through a direct argument.

\subsection{Random affine iterates}
The dynamics~\eqref{eq:gossip-friedkin} can be formally rewritten in vector form 
\begin{align*}
x(k+1)=&(I-\e_i\e_i^\top(I-H))\left(I+\gamma_{ij}(\e_i\e_j^\top-\e_i\e_i^\top)\right)x(k)\\&+\e_i\e_i^\top(I-H)u,
\end{align*}
provided the edge $(i,j)$ is chosen at time $k$.
If we define the matrices 
\begin{align*}A^{(i,j)}=& (I-\e_i\e_i^\top(I-H))\left(I+\gamma_{ij}(\e_i\e_j^\top-\e_i\e_i^\top)\right)\\
B^{(i,j)}=&\, \e_i\e_i^\top(I-H),\end{align*}
then the dynamics is 
$$
x(k+1)=A(k)
x(k) + B(k)
u
$$
where $\Prob[A(k)=A^{(i,j)}]=\frac1{\card{\E}}$
and $\Prob[B(k)=B^{(i,j)}]=\frac1{\card{\E}}$, for all $k\in\integernonnegative$.

\subsection{Expected affine iterates}
The expected dynamics of~\eqref{eq:gossip-friedkin} is 
$$
\Exp[x(k+1)]=
\Exp{[A(k)]}\Exp[x(k)] + \Exp{[B(k)]}u,
$$
and the two average matrices can be explicitly computed as follows, provided Assumption~\ref{assmp:coefficients} holds. 
\begin{align*}
\Exp{[A(k)]}&=\frac{1}{|\E|}\sum_{(i,j)\in \E}A^{(i,j)}\\
&=\frac{1}{|\E|}\sum_{i\in \mathcal{V}}\sum_{j\in \neigh_i}\big[I-\e_i\e_i^\top(I-H)+\gamma_{ij}\left(\e_i\e_j^\top-\e_i\e_i^\top\right)
\\
& \qquad-(I-H)\gamma_{ij}\e_i\e_i^\top\left(\e_i\e_j^\top-\e_i\e_i^\top\right)\big]\\
&=\frac{1}{|\E|}\sum_{i\in \mathcal{V}}\sum_{j\in \neigh_i}\big[I-\e_i\e_i^\top(I-H)+\gamma_{ij}\left(\e_i\e_j^\top-\e_i\e_i^\top\right)
\\
&\qquad
-(I-A)\Gamma_{ij}\left(\e_i\e_j^\top-\e_i\e_i^\top\right)\big]\\
&=\frac{1}{|\E|}\sum_{i\in \mathcal{V}}\sum_{j\in \neigh_i}\left[I-\e_i\e_i^\top(I-H)+H\gamma_{ij}\left(\e_i\e_j^\top-\e_i\e_i^\top\right)\right]\\
&=I-\frac{1}{|\E|}\left[D(I-H)-H\Gamma+H\right].
\end{align*}
Similarly,
\begin{align*}
\Exp{[B(k)]}
&=\frac{1}{|\E|}\sum_{(i,j)\in \E}B^{(i,j)}\\
&=\frac{1}{|\E|}\sum_{i\in \mathcal{V}}\sum_{j\in \neigh_i}\e_i\e_i^\top(I-H)u\\
&=\frac{1}{|\E|}\sum_{i\in \mathcal{V}}|\neigh_i|\e_i\e_i^\top(I-H)u\\&=\frac{1}{|\E|}D(I-H)u.
\end{align*}
%
%
\subsection{Stability of expected iterates}
Before showing the stability of the expected dynamics, which is studied in the Proposition \ref{stability}, we present a technical lemma. 
{ 
Although the result is already known, we prefer to include a short proof for completeness. In order to state the lemma, we need some terminology. The graph {\em associated} to a given square matrix $M\in \real^{\V\times \V}$ is the graph with node set $\V$ and an edge $(i,j)$ if and only if $M_{ij}>0$. We recall that a matrix is said to be substochastic if it is nonnegative and the entries on each of its rows sum up to no more than one. Moreover, every node corresponding to a row which sums to less than one is said to be a {\em deficiency} node. 

\begin{lemma}\label{lemma:substoch_stab}
Consider a  substochastic matrix $M\in \real^{\V\times \V}$. If in the graph associated to $M$ there is a path from every node to a deficiency node, then $M$ is Schur stable.
\end{lemma}

\begin{proof}
 First note that  $M^k$ is substochastic for all $k$.  More precisely,  if we let $V_k$ to be the set of deficiency nodes of $M^k$,  then $V_k\subseteq V_{k+1}$ for every positive integer $k$. Moreover, there exists $k^\star$ such that 
$V_{k^\star}=\V$, that is all nodes for $M^{k^\star}$ are deficiency nodes.
We can then define $\nu=\max_i\sum_j M^{k^\star}_{ij}<1$.
Given any $k\in \N$, we can write $k=n k^\star+r$ with $k\in \{0, \dots, k^\star-1\}$ and integer $n$, and notice that $M^k\1\leq M^{nk^\star}\1\leq \nu^{n}\1$ (provided inequalities are understood componentwise). The last inequality implies that $M^k$ converges to $0$ as $k\to +\infty$.
\end{proof}
 
} 

 
\begin{proposition}\label{stability}
Under Assumptions~\ref{assmp:coefficients} and~\ref{assmp:hm}, the 
matrix $$\bar A=I-\frac{1}{\card{\E}}\big(D(I-H)+H(I-\Gamma)\big)$$ is Schur stable ({\it i.e.}, it has all eigenvalues in the open unit disk).
\end{proposition}

\begin{proof}
 Note that 
$\bar A_{ij}=\frac{1}{\card{\E}} h_i \gamma_{ij}$ if $i\neq j$,  and
$\bar A_{ii}=1-\frac{1}{\card{\E}}\big(   d_i(1-h_i)+h_i(1-\gamma_{ii})\big)$. From these formulas, we observe that all entries of $\bar A$ are nonnegative.
Next, we compute 
\begin{align*}
\sum_{j\in \V} \bar A_{ij}=&1-\frac{1}{\card{\E}}  d_i(1-h_i)+\frac{1}{\card{\E}}h_i(1-\gamma_{ii})+\frac{1}{\card{\E}}  h_i\sum_{j\neq i}\gamma_{ij}\\=&
1-\frac{1}{\card{\E}}  d_i(1-h_i).
\end{align*}

{ Note that $d_i>0$ by the presence of self-loops: consequently $\sum_j \bar A_{ij}<1$ if $h_i<1$.}
Hence, under Assumption~\ref{assmp:hm}, we have that $\bar A$ is a {\em substochastic} matrix corresponding to a graph with a path from any node $\ell$ to a node $m$ whose row sums up to less than one. By Lemma~\ref{lemma:substoch_stab} such a matrix is Schur stable.
\end{proof}
 
 
As a consequence of this result, we deduce that the matrix $D(I-H)+H(I-\Gamma)$ is invertible, and then
\begin{align*}x^{\star}=&\lim_{k\to+\infty}\Exp[x(k)]\\=&(I-\Exp{[A(k)]})^{-1}\Exp{[B(k)]}u\\=&(D(I-H)+H(I-\Gamma))^{-1}D(I-H)u. \end{align*}
{ 
We have by now completed the proof of the first claim of Theorem~\ref{thm:gossip-opinions}.}

\subsection{Ergodicity}
We are now ready to complete the proof of Theorem~\ref{thm:gossip-opinions}, by showing the ergodicity property. Our argument follows the same lines of the convergence results of \cite{HI-RT:10} and \cite{ CR-PF-HI-RT:13a}. Preliminarily, we observe that by the definition of~\eqref{eq:gossip-friedkin}, the opinions $x(k)$ are bounded, as they satisfy \begin{equation}\label{eq:bounded-states}\min_{\ell\in\V}{u_\ell}\le x_i(k)\le\max_{\ell\in\V}{u_\ell}\end{equation} for all $i\in \V$ and $k\ge0$.
In particular, all moments of $x(k)$ are uniformly bounded. 
Let now $e(k):=x(k)-x^{\star}$ be the error from the limit average, and  observe that
$$
\bar x(k)-x^{\star}=\frac{1}{k+1}\sum_{\ell=0}^{k}(x(\ell)-x^{\star})=\frac{1}{k+1}\sum_{\ell=0}^{k}e(\ell).
$$
We thus have
\begin{align*}
\Exp\| \bar x(k)-x^{\star}\|^2&=\Exp\left\| \frac{1}{(k+1)} \sum_{\ell=0}^{k} e(\ell)\right\|^2\\
&=\frac{1}{(k+1)^2} \sum_{\ell=0}^{k}\Exp\left[e(\ell)^{\top} e(\ell)\right]+2 \sum_{\ell=0}^{k} \sum_{\ell=r}^{k-\ell} \Exp \left[e(\ell)^{\top} e(\ell+r)\right].
\end{align*}
In view of~\eqref{eq:bounded-states}, there exists $\eta\in\real$ such that
$$
\frac{1}{(k+1)} \sum_{\ell=0}^{k}\Exp \left[\|e(\ell)\|^2\right] \leq\eta\qquad \forall k.
$$
Next, we note that  
\begin{align}
\nonumber
\Exp \left[e(\ell)^\top e(\ell+r)\right]
& = \Exp\left[ \Exp \left[e(\ell)^\top e(\ell+r)|x(\ell)\right]\right]\\
&  = \Exp\left[ e(\ell)^\top  \Exp \left[e(\ell+r)|x(\ell)\right]\right] \label{eq: x_l}\\
\nonumber&  = \Exp\left[ e(\ell)^\top   \left( \Exp \left[x(\ell+r) |x(\ell)\right]-x^{\star}  \right) \right].
\end{align}
By repeated conditioning on $x(\ell), x(\ell+1), \ldots, x(\ell+r-1),$ we obtain
\begin{align*}
\Exp \big[x(\ell+&r) | x(\ell) \big]=\Exp\left[A(k)\right]^{r} x(\ell)+  \sum_{s=0}^{r-1} \Exp[A(k)]^{s} \Exp[B] u\label{eq: x_l},
\end{align*}
and by recalling that $x^{\star}$ is a fixed point for the expected dynamics we get
\begin{equation}\label{eq: x_star}
x^{\star}=\Exp\left[A(k)\right]^{r} x^{\star}+  \sum_{s=0}^{r-1} \Exp[A(k)]^{s} \Exp[B(k)] u.
\end{equation}
From equations \eqref{eq: x_l} and \eqref{eq: x_star} we obtain
\begin{align*}
\Exp \left[e(\ell)^\top e(\ell+r)\right]&= \Exp\left[e(\ell)^\top  \Exp\left[A(k)\right]^{r} e(\ell)\right]\\
&\leq{\eta}\rho^r,
\end{align*}
where, by Lemma~\ref{stability}, $\rho<1$.
Finally, we have
\begin{align*}
\Exp\left[\| \bar x(k)-x^{\star}\|^2\right]&\leq\frac{\eta}{(k+1)^2}\left(k+1+2\sum_{\ell=0}^{k-1}\sum_{r=0}^{k-\ell}\rho^r\right)\\
&\leq\frac{\eta}{(k+1)}\left(1+\frac{2}{1-\rho}\right),
\end{align*}
from which we obtain the thesis.

\section{Conclusion and open problems}
In this paper, we have defined a new model of opinion dynamics, within the framework of randomized gossip dynamics. We have shown that, for suitable choices of the update parameters, the well-known Friedkin and Johnsen's dynamics is equivalent to the average behavior of the dynamics.
Significantly, the average has a very practical meaning, as the gossip dynamics is ergodic, so that local averages (computed along time) match the expectation.

We note that recent related works on opinion dynamics with stubborn agents~\cite{DA-GC-FF-AO:11,JG-RS:13} have given intuitive characterizations of the (average) limit opinion profile, in terms of harmonic functions or potentials. We leave similar studies on our model as a topic of future research.


\end{document}